\documentclass[12pt]{article}
\usepackage{amsfonts}
\usepackage{amssymb}

\newtheorem{theorem}{Theorem}[section]

\newtheorem{conjecture}[theorem]{Conjecture}

\newtheorem{proposition}[theorem]{Proposition}
\newtheorem{remark}[theorem]{Remark}

\newenvironment{proof}[1][Proof]{\noindent\textbf{#1.} }{\ \rule{0.5em}{0.5em}}
\input{tcilatex}
\begin{document}

\begin{center}
\textbf{Guiding the guiders:\\[0pt]
Foundations of a market-driven theory of disclosure}\\[0pt]
\bigskip

M. Gietzmann, A. J. Ostaszewski, M. H. G. Schr\"{o}der

\bigskip

\bigskip
\end{center}

\noindent \textbf{Abstract.} A foundational approach is developed for a
mathematical theory of managerial disclosure in relation to asset pricing;
this involves both the earnings guidance disclosed by firm management and
market `trackers' pricing the firm's exposure to quotable risks.



\bigskip

\section{Introduction}

Asset-pricing models make assumptions about how information arrives and is
disclosed to its investors (henceforth the market). For assets arising out
of a productive activity by a firm, management reports based on internal
audits of the various \textit{accounting numbers} (accounting variables),
prepared in time for the scheduled (publicly pre-specified) dates, are one
such source. If, by the nature of their activity, management make more
frequent audits (for instance in directing replenishment to a specific
level, which enforces frequent stock-taking, as in the retail business),
then opportunities arise for `early' (unscheduled) disclosure. How should
this additional information be used to signal the firm's superior value,
i.e. to upgrade its share-price? When (or how) should the market `price in'
the absence of early disclosures by a firm to include the possibility that
no news is bad news. The answer must rely on an equilibrium between the
market's ability to down-grade the share-price and the firm's ability to
take advantage of ignorance: hiding some bad news within the uncertain cause
of absent news, i.e. \textit{censoring} the information.

The accounting literature has usually approached this question by including
a \textit{specified} (i.e. known in advance), single, `additional' interim
reporting date, ahead of the next scheduled disclosure, and allowing for
absence of an early disclosure by randomizing the possibility that
management has held an additional audit -- see \cite{Dye}, \cite{JunK}.
However, with the early date a datum, this approach places limits on
comparative analysis; for an equally spaced multi-period model see \cite%
{EinZ}.

\bigskip

\subsection{Earnings guidance: strategic considerations}

In this the companion paper to \cite{GieO} we propose an alternative general
approach, by instilling more realism into the stylized Black-Scholes model
of \cite{GieO}. There the market determination of the (share-)price of the
firms in a sector reflects only the discretional information strategically
released by a firm (i.e. with anticipation of its price effect), usually in
the form of \textit{earnings guidance}, as below, at a stochastic time-point
(i.e. at \textit{unknown} dates in advance of a subsequent mandatory date of
information disclosure). Despite this highly specific origin for the arrival
of information in the market, that model holds considerable advantages,
thanks to its continuous-time approach which overcomes the limitations of
the traditional `two-period analysis' just mentioned. There \textit{%
unspecified} moments in time offer an early disclosure. By way of an example
from \cite{GieO}, which goes beyond the scope intended here, there is a
derivable `band-wagon' effect whereby the introduction of multiple sources
of information reduces an individual firm's optimal frequency of disclosures
by reference to time left to the next mandatory disclosure date.

Typically, however, the market responds also to other public sources of
information, such as trading in the shares of the firm, and by assessing the
exposure of the firm to such economic risks as may be priced by
market-quoted options.

Here we create a more general framework to include such other, already
existing, market-based information enabling the market to make proper use of
this additional information about the firm. This prompts a deeper analysis,
equivalently a formalization at a foundational level, of the various
mechanisms at work, offered in the Complements Section. For simplicity, we
consider here only the situation where the market's concern is for a single
firm, rather than a whole sector. In this we are guided by the clarifying
simplifications that occur in the case of an isolated (`single') firm in the
stylized model \cite{GieO}.

There the firm itself comes to know (`observes') its own state $V_{t}$ only
at discrete, stochastically generated times $t$, not known to the market;
the manager, occupied by a variety of tasks, cannot receive observations
except when these breach agreed thresholds, as reported by personnel
delegated to collect this information, perhaps continuously. This feature of
a hidden observation scheme enables the firm to bury (hide) `bad' news and
only report `sufficiently good' news, principally because the manager cannot
at any time credibly claim as absent an observation at that time. That model
prices the firm in periods of silence, i.e. when the firm fails to supply an
early report of information. Key to this is identifying at each time $t$ a
value $L_{t}$ such that if $L_{t}$ happened to coincide with the true and
currently observed state $V_{t}$, the firm would be \textit{indifferent,} as
regards its market valuation\textbf{\ }(i.e. given the market's information {%
\textit{\ ex ante}}), between choosing to disclose or to withhold the
current observed state $V_{t}$. Such an indifference level $L_{t},$
determined by the information available from before time $t,$ is typically
unique. Censoring, i.e. suppression of an observation below this unique $%
L_{t},$ draws from the market a valuation of the firm at $L_{t}$. In fact, $%
L_{t}$ is the largest possible valuation of the firm, consistent with the
information available from before time $t$. As such it is termed the \textit{%
optimal censor} of time $t.$ Note that observations above $L_{t}$ that are
disclosed cause an upward jump in the firm's valuation. We should emphasize
that \textit{only truthful disclosures} are allowed in the model.

The mathematical argument is based on risk-neutral valuation, which must
incorporate the potential future re-evaluations of firm-value consequent
upon future early disclosures.

The existence of an indifference pricing process is directly attributable to
the firm knowing the market's filtration $\boldsymbol{F}^{\ast }=\{\mathcal{F%
}_{t}^{\ast }\}_{t}$ and the mechanics of how the market performs
computations based on past disclosures (in particular, the probabilities at
each instant which the market attaches to the firm suppressing an
observation of its state). Since the firm's filtration $\boldsymbol{F}=\{%
\mathcal{F}_{t}\}_{t}$ is an enlargement of $\boldsymbol{F}^{\ast }$ \cite%
{Jeu}, in that the firm feeds the market with information by choosing when
to supply its private observations, one may say that the firm \textit{%
emulates} (can simulate) the market. In turn the market's calculations are
based on the firm's algorithm, though not on the firm's up-to-date
observations. The indifference price arises from characterizing a notional
\textit{parametric equilibrium} between the two agents: the firm and the
market (we do not differentiate between informed and noisy traders),
selecting parameters in the computation they use to \textit{second-guess}
each other.

The paper identifies the mechanisms underlying some fairly general valuation
procedures, allowing the market to form its beliefs from additional
information and the firm to exploit the market beliefs by disclosing value
superior to that belief, but nevertheless to give a fair view of future
disclosures. As these mechanisms are inspired by the principal argument and
findings of \cite{GieO}, we close here with a summary of that argument (in
the simplified notation used below). Suppose the next mandatory disclosure
is at the \textit{terminal} time $1$ and that, under the \textit{market's}
risk-neutral measure $Q,$ at time $t<1$ the probability of a disclosure to
the market occurring at time $T\in (t,1]$ is $\tau _{T}.$ (In \cite{GieO}
opportunities to observe the state of the firm are generated by a Poisson
clock.) Based on its information at time $T$, the earnings guidance
announced at that time by the firm gives its \textit{target }terminal\textit{%
\ }value as $VT=E^{Q}[V_{1}|\mathcal{F}_{T}];$ the target and the two
filtrations above are related to the indifference level $L_{T}$ of time $T$
by the two equations%
\[
VT=\tau _{T}E^{Q}[V_{1}|V_{T}\geq L_{T},\mathcal{F}_{T}]+(1-\tau _{T})L_{T},%
\leqno{(1.1.1a)}
\]%
\[
L_{T}=E^{Q}[V_{1}|ND_{T}(L_{T}),\mathcal{F}_{T}^{\ast }].\leqno{(1.1.1b)}
\]%
Here $ND_{T}(L)$ is the event that no disclosure occurs at time $T,$ which
means that either there has been no opportunity to observe $V_{T}$ or else
the manager has observed $V_{T}$ but $V_{T}\leq L.$ From here, in the
context of \cite{GieO}, given how $\mathcal{F}_{T}^{\ast }$ is generated
from $\mathcal{F}_{T}$ \textbf{via} $\mathcal{F}_{T}^{\ast }$-measurable
decisions, one deduces in the limit as $T\rightarrow t$ from (1.1.1a) that $%
t\mapsto L_{t}$ satisfies a simple ordinary differential equation (involving
the instantaneous variance of $Q$ and the Poisson clock's intensity).
Assuming multiplicative \textit{scalability,} that $V_{T+u}=V_{T}\tilde{V}%
_{u}$ with independence of $V_{T}$ and $\tilde{V}_{u}$, equations (1.1.1a,b)
can be further simplified.

In summary, apart from simplifications, this paper's contributions include:
announcements of both sufficiently good and sufficiently bad news (dual,
`materiality' aspects in the release of private information), incorporation
of current public information in modelling market sentiment, and comparative
statics of early disclosures.

The paper is organized as follows. Section 2 contains a preliminary
discussion of our modelling aims. Section 3 models the market's beliefs as
to firm value, based on expected performance indices and is supported by a
geometric Brownian (GMB)\ implementation. This is followed in Section 4 by a
model of the firm setting its target values; using a benchmark scheme to be
followed by the firm in observing its own state, this is shown at its
simplest to be similar to determining option exercise values, and is
supported by an indicative GMB\ implementation. This enables us to perform
comparative statics in Section 5, which we conclude meets a primary
objective: to show how parameter values determine early or delayed voluntary
release of information in equilibrium. We comment briefly on the
implications of our approach in Section 6, thus rounding off the paper's
contribution, and close in Section 7 with Complements indicating a framework
for generalizations and potential variations. An Appendix gives details of
well-known GMB\ formulas needed in the paper.

\section{Model preliminaries}

In this section we introduce a model of how a firm $F$ decides to disclose
information intermittently at times $t$ to the market $M$ about its state $%
V_{t},$ voluntarily between legally mandated (mandatory) disclosure dates.
This involves modelling how the market forms beliefs $V_{t}^{\ast }$ in
periods of silence about the true current value $V_{t}.$ (We regard the
market as dual to the firm, hence the `star' notation here and below.)

Our first two tasks are: to model the beliefs of $M$ (in \S 3) and then to
model $F$'s choice of `equilibrium' \textit{indifference level} (in \S 4),
below which an observation of $V_{t},$ if any, is not disclosed (as in the
Introduction). We will rely on tractable Black-Scholes frameworks, and in
the second task we will be guided by the findings of \cite{GieO}. In \S 5 we
prove the existence of these indifference levels in a benchmark context. We
may then pass to calculations which will yield conclusions, in particular,
about the likelihood of early disclosure. This enables us in \S 6 to address
comparative statics of voluntary (i.e. early) disclosures, matters beyond
the reach of \cite{GieO}.

Thereafter, we discuss generalizations identifying potential for more
sophisticated models (e.g. inclusion of analyst forecasts).

The firm has `private' access, according to some\textit{\ observation scheme}
-- possibly also intermittent -- to its `state' $V_{t}$ (thought of as the
income stream). This is modelled as a random time $\tau $ not known to $M.$
The firm applies fixed decision rules by which it determines at time $t$
whether to withhold any observation it may have, or to disclose its
information voluntarily (and truthfully) to the market via two items of
information: (i) the current value $V_{t},$ and (ii) the expected state at
the terminal date, i.e. the next mandatory disclosure date. We term the
former the declared\textit{\ current value} $VC$ and the latter its declared
\textit{target} value $VT$ (current as at the date of its disclosure). The
firm's intention is to achieve the highest possible market valuation at each
point in time; here this is implemented by use of a fixed \textit{decision
rule}, based both on its own private information about its state and on the
market's public belief about the state of the firm, which in turn depends on
the market's information base. We term this, publicly held belief about the
value, the \textit{market sentiment}. $F$ forms its expectations by
reference to a measure $Q_{VT}$ (labelled by the last declared target) under
which its (discounted) observation process is a martingale.

To model market sentiment, we borrow and amend a concept from control
theory, that of a \textit{state observer} (a.k.a. `state estimator') system;
for a discussion see \S 7.1(iv).

In the \cite{GieO} model the observation scheme was a Poisson clock with
known jump intensity; the market sentiment was an equilibrium valuation
obtained from the latest disclosed state-value, prudently discounted
downwards (i.e. by incorporating the possible undisclosed poor performance);
discounting is by a rate determined by the (known) Poisson jump intensity.

Below, the market sentiment $V_{t}^{\ast }$ is based on the latest disclosed
information and on the current value, or perhaps on the prevailing
behaviour, of some specified portfolio of traded assets with value $%
S_{t}^{\ast }$ (e.g. current value, average value, record value to date).
The portfolio, termed the \textit{tracker}, is viewed by the market as
capturing the firm's exposure to quotable (market-priced) risks. The key
property of $S_{t}^{\ast }$ is that it is priced by a risk-neutral (i.e.
market) measure. That is, $M$ forms its expectations by reference to a
measure $Q^{\ast }$ under which the (discounted) tracker process is a
martingale. We note that, at each time $t,$ since $F$ has access to $M$'s
information $\mathcal{F}_{t}^{\ast }$ plus its own observation, i.e. its own
filtration $\{\mathcal{F}_{t}\}_{t}$ is an enlargement of $\{\mathcal{F}%
_{t}^{\ast }\}_{t}$,
\[
Q^{\ast }|\,\mathcal{F}_{t}^{\ast }=Q_{VT}|\,\mathcal{F}_{t}^{\ast }.
\]

As in \cite{GieO}, so too here, the link between market sentiment $V^{\ast }$
and asset $S^{\ast }$ needs to be determined by equilibrium considerations:
if at time $t$ the firm applies a decision rule $h$ to the observations of $%
S^{\ast }$ and $V$, it will determine an \textit{indifference level} $L$ (as
in the Introduction) which, if $F$ observed that $V_{t}=L,$ would make $F$
indifferent between disclosure or otherwise of $V_{t}$. While the firm
determines its disclosure using a rule $h$ (below) that exploits any
superiority of the observed value over market sentiment, rather than its
expected terminal value, it complements such a disclosure by supplying the
market with information about the expected terminal value.

We take the decision rules for $F\ $and consequently also for $M$ (with
starred notation for the latter's rules) in the (time-independent) form
\[
h_{\varepsilon ,a}(t,x,y)=(x-(1\!+\!a)y)\varepsilon ,\qquad (x,y\in \mathbb{R%
},t>0)\qquad \leqno{(2.0.1)}
\]%
where $a>-1$ is termed a \textit{mark-up}, and $\varepsilon =\pm 1$ its
\textit{signature}, positive when deciding a \textit{good-news }disclosure
event at time $u,$ say, when
\[
V_{u}\geq (1\!+\!a)V_{u}^{\ast }\,,\leqno{\rm (2.0.1a)}
\]%
and respectively negative for a disclosure of \textit{bad-news}. (A dynamic
variant is considered briefly in \S 7.1) These may be viewed as backed by a
theoretical justification for such a `principal-agent' delegation (here a
shareholder-mandated policy): see the classic paper {\cite[Prop. 1.4]{BaiD}}
for a rigorous derivation of control limits, using an `evaluation and
control' method. The argument there refers to the costs versus the benefits
of extracting information and the authors claim support of (perhaps,
anecdotal) hard evidence that such rules are observed in practice.

We close by stressing that the various asset-price dynamics here are
modelled only between consecutive disclosures -- in `periods of silence'.

\section{Market sentiment: shadowing the firm}

Our first task begins at a point in time $T_{0}$ with $0\leq T_{0}\leq 1$
when the firm is assumed to have disclosed two items of data: current state $%
VC$ and declared target $VT.$ We use $VT$ as a suffix conveniently labelling
the various processes started at time $T_{0}$ .

The next time of disclosure, following access to an observation of $X_{T}$
at time $T=T_{+},$ will occur provided
\[
h(T,V_{T},V_{T}^{\ast })\geq 0.
\]%
Here the firm $F$ applies its \textit{decision rule }$h$ taken in general to
depend on the time-$T$ values of $V$ and $V^{\ast },$ and perhaps on $T$
itself, a possibility ruled out below to simplify calculations (hence the
rules in \S 2 above). Then at time $T=T_{+},$ the firm will declare its
current state $V_{T}$ and set a new target value $VT=VT_{+}.$ So the main
tasks are to devise a justifiable model for a process $V^{\ast },$ which we
view as `shadowing the firm', and for $VT_{+}$ (in the next section).

\subsection{Two modelling assumptions}

We begin by identifying \textit{how} to model $V_{t}^{\ast }.$ This will be
determined by two components. Although the entire process is driven by a
specified market portfolio $S^{\ast }$, the \textit{tracker}, nevertheless,
at each non-disclosure time-point $t$ in $(T_{0},T_{1})$ a correction term
needs to be included to price in the effects of any possible future
disclosure event, say of time $\tau $. We thus aim for a two-component model
\[
V_{t}^{\ast }=S_{t}^{\ast }+\Delta V_{t}^{\ast }\,,\leqno{(3.1.1)}
\]%
which requires that we model $\Delta V_{t}^{\ast }$. The latter must refer
to contingent claims in regard to whether or not $F$ has attained its target
to date; with $h^{\ast }$ a market decision rule, this requires pricing
contingent claims on $[t,\tau ]$ characterized by the instantaneous time-$%
\tau $ pay-off
\[
S_{\tau }^{\ast }\mathbf{1}_{\{h^{\ast }(\tau ,S_{\tau }^{\ast },VT)\geq
0\}},\quad \text{for}\quad \tau =\tau _{VT}\,.\leqno{(3.1.2)}
\]%
So this is a `securitization' similar to standard derivative instruments. It
depends on a `random time' $\tau =\tau _{VT}$ (with values in $(t,T_{1}]$)
unknown to the agent $M$. Unfortunately, pricing these depends on \textit{%
individual} investor attitudes, so on micro level information, typically
unobservable. The pragmatic approach is to replace the pricing of these
claims with approximations. Assuming a constant risk-free rate $r$ in force,
we now make explicit a \textit{first modelling assumption}, that with $\tau
=\tau _{VT}$
\[
\Delta V_{t}^{\ast }:=\exp (-r(T_{1}\!-\!t))E^{Q^{\ast }}[S_{\tau }^{\ast }\,%
\mathbf{1}_{\{h^{\ast }(\tau ,S_{\tau }^{\ast },VT)\geq 0\}}\,|\,\mathcal{F}%
_{t}^{\ast }\,]-S_{t}^{\ast };\leqno{(3.1.3)}
\]%
here $Q^{\ast }$ is an assumed risk-neutral measure conditional on market
information -- conditional at time $t$ on market information $\mathcal{F}%
_{t}^{\ast }$ (with the expectation on the right assumed meaningful). The
formula identifies $\Delta V_{t}^{\ast }$ as the excess over $S_{t}^{\ast }$
of the fair value of the effect of a disclosure occurring at time $\tau .$
The final step at time $t$ is to pass to an approximation, which make
explicit a \textit{second modelling assumption}, that
\[
S_{\tau }^{\ast }\,\mathbf{1}_{\{h^{\ast }(\tau ,S_{\tau }^{\ast },VT)\geq
0\}}\approx E_{t}^{\ast }\,\mathbf{1}_{\{h^{\ast }(t,\Sigma _{t,T_{1}}^{\ast
},VT)\geq 0\}}.\leqno{(3.1.4)}
\]%
Here, on the one hand, $\Sigma _{t,T_{1}}^{\ast }$ is some chosen \textit{%
tracker-performance index} over the entire remaining time interval, for
instance%
\[
\left.
\begin{array}{c}
\Sigma _{t,T_{1}}^{\ast }=\max \{S_{u}^{\ast }|u\in \lbrack t,T_{1}]\},\text{
or} \\
\Sigma _{t,T_{1}}^{\ast }=1/(T_{1}-t)\int_{[t,T_{1}]}S_{u}^{\ast }\text{ }%
\mathrm{d}u,\text{ or} \\
\Sigma _{t,T_{1}}^{\ast }=\min \{S_{u}^{\ast }|u\in \lbrack t,T_{1}\};%
\end{array}%
\right\} \leqno{(3.1.5abc)}
\]%
and, on the other hand, a typical choice for $E_{t}^{\ast }$ is obtained by
solving $h^{\ast }(t,E_{t}^{\ast },VT)=0$. This last has a unique solution,
the result of the simple form of decision rule in \S 2 above. Thus in
good-news situations $E_{t}^{\ast }$ identifies a value \textit{%
at-least-as-good-as \/}the true value $S_{\tau }^{\ast }$ at disclosure; in
bad-news situations such an $E_{t}^{\ast }$ will be \textit{%
at-most-as-bad-as\/ }that. We summarize these modelling considerations in:

\bigskip

\begin{proposition}
\quad \textit{The two modelling assumptions }(3.1.3), (3.1.4)\textit{\ imply
that
\[
V_{t}^{\ast }=(S_{t}^{\ast }+\Delta E_{t}^{\ast })\,Q^{\ast }(h^{\ast
}(t,\Sigma _{t,T_{1}}^{\ast },VT)\geq 0\,|\,\mathcal{F}_{t}^{\ast })\exp
(-(T_{1}\!-\!t)r)\,,
\]%
where:\newline
}\noindent (i)\textit{\ $S^{\ast }$ is assumed to start with the value $VC$
at time $T_{0}$,\newline
}\noindent (ii) \textit{$\Delta E_{t}^{\ast }=E_{t}^{\ast }\!-\!S_{t}^{\ast
} $, and\newline
}\noindent (iii)\textit{\ $\Sigma _{t,T_{1}}^{\ast }$ is chosen as in
example }(3.1.5)\textit{\ above.}
\end{proposition}

\bigskip

\subsection{A geometric Brownian implementation}

On the firm side we take the firm's observation process $V$ to be modelled
by a geometric Brownian motion $X$:
\[
X_{s+t}=X_{t}\exp ((\mu _{VT}-\hbox{$ {1\over 2}$}\sigma _{VT}^{2})s+\sigma
_{VT}W_{VT,s}),\quad s\in \lbrack 0,T\!-\!t];\leqno{(3.2.1)}
\]%
here $W_{VT}$ is a Brownian motion independent of time-$t$ information $%
\mathcal{F}_{t}$, with $\sigma _{VT}>0$ and $\mu _{VT}\in \mathbb{R}$ fixed.

Likewise, on the market side we adopt a Black-Scholes model with one risky
security $S^{\ast },$ so modelled again by a geometric Brownian motion,
which thus satisfies the requirements of Proposition~3.1 and carries several
advantages beside: firstly, model-completeness (see for example {\cite[%
Prop.~8.2.1, p.~302]{MusR}}  for the completeness of the multi-dimensional
model) and, secondly, the ability of being re-started with value\textbf{\ }$%
VC$ at time $T_{0}.$ We take the dynamics in the form
\[
S_{t+u}^{\ast }=S_{t}^{\ast }\exp (\mu ^{\ast }u+\sigma ^{\ast }W_{u}^{\ast
}),\quad u\in \lbrack 0,\infty )\,\leqno{(3.2.2)}
\]%
with Brownian motion $W^{\ast }$ independent of time-$t$ information $%
\mathcal{F}_{t}^{\ast }$, and two parameters: volatility $\sigma ^{\ast }>0$
and drift $\mu ^{\ast }=r-\delta -(1/2){\sigma ^{\ast }}^{2}$, for some $%
\delta \in \mathbb{R}$.

This modelling choice ensures that the probabilities in the formula of
Proposition~3.1 are well-defined (see below), emerging as tail probabilities
for good-news decisions, since
\[
h^{\ast }(t,\Sigma ^{\ast },VT)\geq 0\quad \hbox{iff}\quad h_{+1,a^{\ast
}}(\Sigma ^{\ast },VT)\geq 0\quad \hbox{iff}\quad \Sigma ^{\ast }\geq
(1\!+\!a^{\ast })VT,
\]%
for $\Sigma ^{\ast }$ any random variable. It is routine to derive explicit
formulae for these, see the next subsections. We focus here on the choices
of (3.1.5a,c) of max- and min-tracker-performance index $\Sigma
_{t,T_{1}}^{\ast }(S^{\ast })$, leaving aside the average value index.

\subsubsection{\textit{\ V}$^{\ast }$ from \textit{the running-max
approximation}}

We first deal with the running-max
\[
\Sigma ^{\ast }:=\max \{S_{u}^{\ast }\,|\,u\in \lbrack t,T_{1}]\}.
\]%
We note a reversion from bad-news to good-news\textbf{\ }decisions via:
\[
Q^{\ast }(h_{+1,a^{\ast }}(\Sigma ^{\ast },VT)\geq 0)+Q^{\ast
}(h_{-1,a^{\ast }}(\Sigma ^{\ast },VT)>0)=1.
\]%
The actual influence of the second term on the first is determined by the
relative size of $(1\!+\!a^{\ast })VT$ and $S_{t}^{\ast }$, so according to
the sign of
\[
A^{\ast }=\log ((1\!+\!a^{\ast })VT/S_{t}^{\ast }).\leqno{(3.2.3)}
\]%
Below $\mathrm{Erfc}\,$is the complementary error function, for which see
subsection A2 of the Appendix. Equations (A.7a) and (A.7b) in the Appendix
yield the following results.

\bigskip

\begin{proposition}
\quad \textit{If }$A^{\ast }<0$\textit{, equivalently, }$S_{t}^{\ast
}>(1\!+\!a^{\ast })VT$\textit{, then}%
\[
Q^{\ast }(h_{-1,a^{\ast }}(\Sigma ^{\ast },VT)>0)=0,
\]%
\textit{\ and there is no `bad-news' influence; thus}
\[
Q^{\ast }(h_{+1,a^{\ast }}(\Sigma ^{\ast },VT)\geq 0)=1\,,\quad \text{if }%
S_{t}^{\ast }\geq (1+a^{\ast })VT.\leqno{\rm (3.2.4a)}
\]%
\textit{If }$A^{\ast }\geq 0$\textit{, equivalently }$S_{t}^{\ast }\leq
(1\!+\!a^{\ast })VT$\textit{, then `bad-news' carries influence measured by}%
\[
Q^{\ast }(h_{+1,a^{\ast }}(\Sigma ^{\ast },VT)\geq 0)=1-Q^{\ast
}(h_{-1,a^{\ast }}(\Sigma ^{\ast },VT)>0),\leqno{\rm(3.2.4b)}
\]%
\textit{where}%
\[
Q^{\ast }(h_{-1,a^{\ast }}(\Sigma ^{\ast },VT)\geq 0)=\leqno{\rm(3.2.4c)}
\]%
\[
=\frac{1}{2}\mathrm{Erfc}\left( \frac{A_{t}^{\ast }-(T-t)\mu ^{\ast }}{%
\sigma ^{\ast }\sqrt{2(T-t)}}\right) +\frac{1}{2}\exp \left( \frac{2\mu
^{\ast }A_{t}^{\ast }}{(\sigma ^{\ast })^{2}}\right) \mathrm{Erfc}\left(
\frac{A_{t}^{\ast }+(T-t)\mu ^{\ast }}{\sigma ^{\ast }\sqrt{2(T-t)}}\right) .
\]
\end{proposition}

\bigskip

To complete the picture, note that the bad news scenario when $\Sigma ^{\ast
}$ is the running-maximum can be read back from:
\[
Q^{\ast }(h_{-1,a^{\ast }}(\Sigma ^{\ast },VT)\geq 0)=0\,,\quad \quad \text{%
if }S_{t}^{\ast }>(1+a^{\ast })VT,\leqno{\rm (3.2.5)}
\]%
and that for $A^{\ast }>0$ this probability is given by equation (3.2.4c)
above.

\subsubsection{\textit{\ V}$^{\ast }$ from the \textit{running-min
approximation}}

We now deal with running-min
\[
\Sigma ^{\ast }:=\min \{S_{u}^{\ast }\,|\,u\in \lbrack t,T_{1}]\}.
\]%
The good-news and bad-news formulas hold good and, viewed technically, may
be derived by replacing $\mu ^{\ast }$ by $-\mu ^{\ast }$ and $A^{\ast }$ by
$-A^{\ast }=\log (S_{t}^{\ast }/((1\!+\!a^{\ast })VT))$; and switching to
probabilities complementary to those in (3.2.1): see the discussion for
equation (A8) in Section~A2 of the Appendix. From there, we have explicitly:

\bigskip

\begin{proposition}
\quad \textit{If }$S_{t}^{\ast }\leq (1\!+\!a^{\ast })VT,$\textit{\ then}%
\[
Q^{\ast }(h_{+1,a^{\ast }}(\min \{S_{u}^{\ast }\,|\,u\in \lbrack
t,T_{1}]\},VT)\geq 0)=0,\leqno{\rm (3.2.6a)}
\]%
\[
Q^{\ast }(h_{-1,a^{\ast }}(\min \{S_{u}^{\ast }\,|\,u\in \lbrack
t,T_{1}]\},VT)\geq 0)=1.\leqno{\rm (3.2.6b)}
\]%
\textit{If }$S_{t}^{\ast }>(1\!+\!a^{\ast })VT,$\textit{\ then}
\[
\leqno{\rm (3.2.6c)}\qquad Q^{\ast }(h_{+1,a^{\ast }}(\min \{S_{u}^{\ast
}\,|\,u\in \lbrack t,T_{1}]\},VT)\geq 0)
\]%
\[
=1-Q^{\ast }(h_{-1,a^{\ast }}(\min \{S_{u}^{\ast }\,|\,u\in \lbrack
t,T_{1}]\},VT)>0),
\]%
where
\[
\leqno{\rm(3.2.6d)}\qquad Q^{\ast }(h_{-1,a^{\ast }}(\min \{S_{u}^{\ast
}\,|\,u\in \lbrack t,T_{1}]\},VT)\geq 0)=
\]%
\[
=\frac{1}{2}\mathrm{Erfc}\left( \frac{A_{t}^{\ast }-(T-t)\mu ^{\ast }}{%
\sigma ^{\ast }\sqrt{2(T-t)}}\right) -\frac{1}{2}\exp \left( \frac{2\mu
^{\ast }A_{t}^{\ast }}{(\sigma ^{\ast })^{2}}\right) \mathrm{Erfc}\left( -%
\frac{A_{t}^{\ast }+(T-t)\mu ^{\ast }}{\sigma ^{\ast }\sqrt{2(T-t)}}\right) .
\]
\end{proposition}

\section{Setting new targets}

The previous section determined how the firm triggers disclosure by
reference to a fixed decision rule $h$ and a model of market sentiment $%
V_{t}^{\ast }$ (our proxy for an observer-system of control theory). This
replaces and simplifies the dynamics of the equilibrium approach of \cite%
{GieO}, but comes at the cost of losing information about the expected
terminal firm-value (value at the next mandatory disclosure date). The model
of \cite{GieO} identifies that expected `terminal value' as equal to the
disclosed value. The modelling in the current section provides the missing
information in the form of a new \textit{target value} $VT_{+}$ (so `plugs'
the gap between the models). The framework here is, nevertheless, inspired
by the equilibrium argument in \cite{GieO}, as summarized by the concept of
an indifference level (see equations (1.1.1a,b) in the Introduction).

If the firm were to use a threshold $L$ to trigger disclosure at some, for
the moment arbitrary, future time moment $u$ in $(t,T]$, the firm's adopted
decision rule, $h$ say, determines disclosure iff $h_{VT,u}(u,X_{u},L)\geq
0. $ As only truthful disclosures are assumed, this entails a market
valuation at the disclosed level. However, absence of a disclosure entails,
for some appropriately selected threshold $L,$ as in the model of \cite{GieO}%
, a valuation of $L.$ In summary, if $L=L(u)=L_{u}$ is selected
appropriately for time $u$, then the time $u$ valuation of the firm is given
by the random variable
\[
Z_{VT,u}(L_{u})=X_{u}\,\mathbf{1}_{\{h_{VT,u}(u,X_{u},L(u))\geq 0\}}+L\,%
\mathbf{1}_{\{h_{VT,u}(u,X_{u},L(u))<0\}}.\leqno{ (4.0.1a)}
\]

Now let $\tau _{VT}$ be a random time, with the interpretation that the
event $\tau _{VT}(u)=u$ for $u>0$ means that $F\ $observes $X_{u},$ the
complementary event being $\tau _{VT}(u)=0.$

We now modify the random variables in (4.0.1a) by taking into account the
times of observation and non-observation and define%
\[
Z_{VT,u}^{\tau _{VT}}(L_{u})=Z_{VT,u}(L)\,\mathbf{1}_{\{\tau
_{VT}(u)=u\}}+L\,\mathbf{1}_{\{\tau _{VT}(u)=0\}}.\leqno{ (4.0.1b)}
\]%
Then the expected valuation is
\[
\int_{(t,T]}E^{Q_{VT}}[Z_{VT,u}^{\tau _{VT}}(L_{u})\,|\,\mathcal{F}%
_{t}\,]\,\tau _{VT}(\mathrm{d}u),
\]%
denoting here the distribution of $\tau _{VT}$ by $\tau _{VT}$ again, for
notational convenience. As in \S 1.1 (cf. \cite{GieO}), since $Q_{VT}$ is
risk-neutral, this should agree with $X_{t}.$ Without loss of generality to
the analysis of the interval $(t,T],$ we may agree to resize (rescale) the
observation process at time $t$ to unity. Interpreting the values as
discounted to present time $t$, our modelling assumption is to seek a
constant $L=L_{VT}$ solving%
\[
1=\int_{[t,T]}E^{Q_{VT}}[Z_{VT,u}^{\tau _{VT}}(L_{VT})\,|\,\mathcal{F}%
_{t}\,]\,\tau _{VT}(\mathrm{d}u).\leqno{(4.0.2)}
\]%
In setting the new target level, this formula relies not on the market
filtration $\mathbf{F}^{\ast }$ (so \textit{not} on future market
sentiment), but on fair value computed from the larger filtration $\mathbf{F}
$ with which the firm is equipped.

Granted the existence of a solution to (4.0.2), a matter addressed in \S 4.2
below, we take $VT_{+}:=L_{u}$ to correspond to $h_{VT,u}.$

\subsection{A bench-mark observation scheme}

For a tractable implementation of the modelling assumption in formula
(4.0.2), we replace the random observation scheme $\tau $ by a deterministic
one, known only to the firm but most certainly not known to the market. This
permits a decomposition%
\[
\lbrack t,T]=\mathcal{C}_{VT}\cup \mathcal{D}_{VT}\cup \mathcal{N}_{VT}%
\leqno{(4.1.1)}
\]%
according as observation extends over continuous intervals, or either at a
finite number of (discrete) time moments, or not at all.

Then the equation above reduces to
\[
1=\int_{\mathcal{C}_{VT}}E^{Q_{VT}}[Z_{VT,u}^{\tau _{VT}}(L_{u})\,|\,%
\mathcal{F}_{t}\,]\,\frac{\mathrm{d}u}{T-t}
\]%
\[
\qquad \qquad +\sum\nolimits_{u\in \mathcal{D}_{VT}}E^{Q_{VT}}[Z_{VT,u}^{%
\tau _{VT}}(L_{u})\,|\,\mathcal{F}_{t}\,]\,q_{VT}\text{ }\leqno{(4.1.2)}
\]%
\begin{equation}
+\left( 1-\frac{\mathrm{vol}(\mathcal{C}_{VT})}{T-t}\right) L_{VT}  \nonumber
\end{equation}%
with $q_{VT}=1/\#\mathcal{D}_{VT}$, effectively the constant probability of
discrete monitoring.

Furthermore, taking $h_{VT,u}$ to be $h_{\varepsilon ,a_{VT,u}}$ (with a
mark-up $a_{VT,u}>-1$) leads to a decomposition of the variable $Z_{VT,u}$
into an option part and a non-option part, appropriately corresponding to
good-news and bad-news events. For the good-news case ($\varepsilon =+1)$,
it can readily be checked that this is
\[
Z_{VT,u}(L)=L+\max \{X_{u}-(1+a_{VT,u})L,0\}+a_{VT,u}L\mathbf{1}%
_{\{X_{u}\geq (1+a_{VT,u})L\}},\leqno{\rm(4.1.3a)}
\]%
and similarly for bad-news ($\varepsilon =-1):$%
\[
Z_{VT,u}(L)=L-\max \{(1+a_{VT,u})L-X_{u},0\}+a_{VT,u}L\mathbf{1}%
_{\{X_{u}\leq (1+a_{VT,u})L\}}.\leqno{\rm(4.1.3b)}
\]%
So the `optionality' in $Z_{VT,u}(L)$ reduces to a \textit{plain vanilla
option} corrected by a \textit{digital option}. Turning to the
practicalities of options, one way to handle positions in digital options is
to approximate them by plain vanilla positions using a selection of slightly
amended strikes. From this perspective, the optionality of $Z_{VT,u}(L)$ can
be regarded as approximately induced by a plain vanilla call- (respectively
put-)\ option with strikes close to $(1\!+\!a_{VT,u})L$. In view of its
broader role we refer to $L_{VT}$ as the \textit{optimal censor} (cf. \S %
1.1).

\bigskip

\begin{proposition}
\textbf{\ (Optimal censor optionality):}\quad \textit{When }$a_{VT,u}=0$%
\textit{\ for all }$u,$\textit{\ with the additional assumption of only
discrete observations} ($\mathcal{C}_{VT}=\emptyset $),\textit{\ the
equation (4.1.2) for the optimal censoring thresholds specializes for the
good-news event to}%
\[
1=2L_{VT}+\,q_{VT}\sum\nolimits_{u\in \mathcal{D}_{VT}}E^{Q_{VT}}[\max
\{X_{u}-L_{VT},0\}\,|\,\mathcal{F}_{t}\,],\leqno{\rm(4.1.4a)}
\]%
\textit{and for bad-news}%
\[
1=2L_{VT}-\,q_{VT}\sum\nolimits_{u\in \mathcal{D}_{VT}}E^{Q_{VT}}[\max
\{L_{VT}-X_{u},0\}\,|\,\mathcal{F}_{t}\,].\leqno{\rm(4.1.4b)}
\]
\end{proposition}

\bigskip

\subsection{Existence of $L_{VT}$ for the benchmark observation scheme}

This section demonstrates the existence of a target value $VT_{+}:=L_{VT}$
for the benchmark observation scheme of the preceding subsection as
characterized by equation (4.1.2). The existence theorems splits into two
cases according as the decision rule determines good- or bad-news
announcements; in both cases we analyze the functional form on the right of
equation (4.1.2), treating $L_{VT}$ as a free variable, now denoted by $L.$
It is convenient to begin with bad-news announcements.

\subsubsection{Bad-news case}

The final term in equation (4.1.2), corresponding to non-moni\-tor\-ing, has
a simple functional form: it is linear in $L.$ To understand the other
contributions, we rewrite the equation in a form which reflects the
complementary conditioning in the two summands of the earlier equation
(4.0.1b). This gives rise below to two corresponding functions of $L$ and
recasts the characterization of $L_{VT}$ in the form:%
\[
1=\mathcal{N}_{VT}(L_{VT})+\mathcal{BS}_{VT,1}(L_{VT})+\mathcal{BS}%
_{VT,2}(L_{VT}).\leqno{\rm(4.2.1a)}
\]%
The three functions appearing here are defined as follows:%
\[
\mathcal{N}_{VT}(L)=\left( 1-\frac{\mathrm{vol}(\mathcal{C}_{VT})}{T-t}%
\right) L,\leqno{\rm(4.2.1b)}
\]%
\[
\mathcal{BS}_{VT,1}(L)=\int_{\mathcal{C}_{VT}}E^{Q_{VT}}[X_{u}\mathbf{1}%
_{\{X_{u}\leq (1+a_{VT,u})L\}}\,|\,\mathcal{F}_{t}\,]\,\frac{\mathrm{d}u}{T-t%
}
\]%
\begin{equation}
\qquad +\,q_{VT}\sum\nolimits_{u\in \mathcal{D}_{VT}}E^{Q_{VT}}[X_{u}\mathbf{%
1}_{\{X_{u}\leq (1+a_{VT,u})L\}}\,|\,\mathcal{F}_{t}\,],\leqno{\rm(4.2.1c)}
\nonumber
\end{equation}%
\[
\mathcal{BS}_{VT,2}(L)/L=\int_{\mathcal{C}_{VT}}E^{Q_{VT}}[\mathbf{1}%
_{\{X_{u}>(1+a_{VT,u})L\}}\,|\,\mathcal{F}_{t}\,]\,\frac{\mathrm{d}u}{T-t}
\]%
\begin{equation}
\qquad +\,q_{VT}\sum\nolimits_{u\in \mathcal{D}_{VT}}E^{Q_{VT}}[\mathbf{1}%
_{\{X_{u}>(1+a_{VT,u})L\}}\,|\,\mathcal{F}_{t}\,]\leqno{\rm(4.2.1d)}
\nonumber
\end{equation}%
(with `B for bad news' and `S\ for Black-Scholes'). The relation between
their behaviour and consequent existence of a target value is captured in
the following result.

\bigskip

\begin{proposition}
\quad \textit{In bad-news events, for a constant $L_{VT}$ to exist for which
equation $(4.2.1a)$ holds, the following conditions $(1)$ to $(4)$ are
sufficient. }

\begin{itemize}
\item[\textrm{(1)}] $\mathcal{BS}_{VT,1}(L)$ \textit{\ and $\mathcal{BS}%
_{VT,2}$ are continuous maps on $[0,\infty )$.}

\item[\textrm{(2)}] \textit{\ }$\mathcal{BS}_{VT,1}$\textit{$(\infty
)>-\infty $. }

\item[\textrm{(3)}] \textit{\ $\mathrm{vol}(\mathcal{C}_{VT})\neq T\!-\!t$,
or $\mathcal{BS}_{VT,2}$ is unbounded. }

\item[\textrm{(4)}] \textit{\ $1\geq \mathcal{BS}_{VT,1}(0)$. }\medskip
\end{itemize}
\end{proposition}

\begin{proof}
First consider the behaviour of the function
summands as $L$ grows large. In $\mathcal{BS}_{VT,1}(L)$ the
indicator-functions for large $L$ will become those of the entire space,
i.e. the constant function $1$; the summands of $\mathcal{BS}_{VT,1}(L)$
should thereby become expressible in terms of the first moments of $X$ as
follows:
\begin{equation*}
\mathcal{BS}_{VT,1}(L)(\infty )=\int\limits_{\mathcal{C}%
_{VT}}E^{Q_{VT}}[X_{u}\,|\,\mathcal{F}_{t}\,]\,{\frac{\mathrm{d}u}{T\!-\!t}}%
+q_{VT}\sum\limits_{u\in \mathcal{D}_{VT}}E^{Q_{VT}}[X_{u}\,|\,\mathcal{F}%
_{t}\,].\text{ }\leqno{(4.2.1c})_{\infty }
\end{equation*}%
Little can be said about the behaviour for large $L$ of $\mathcal{BS}%
_{VT,2}(L)$ except for its being non-negative for $L$ non-negative. As a
consequence,
\begin{equation*}
\mathcal{N}_{VT}(L)+\mathcal{BS}_{VT,1}(L)+\mathcal{BS}_{VT,2}(L)\geq
\mathcal{N}_{VT}(L)+\mathcal{BS}_{VT,1}(L)\,,
\end{equation*}%
for any $L\geq 0$. On inspection from (4.2.1b), the right hand side of this
inequality will grow linearly in $L$ arbitrarily provided $\mathrm{vol}(%
\mathcal{C}_{VT})/(T\!-\!t)\neq 1$. Situations where $\mathrm{vol}(\mathcal{C%
}_{VT})=T\!-\!t$ amount to monitoring $X$ at all points in time in $[t,T]$
except perhaps on an infinite sequence of points; this is contrary to the
spirit of this paper's observation schemes $\tau _{VT}$, and so little will
be lost in excluding such schemes. A minor problem arises, when $\mathcal{BS}%
_{VT,1}(\infty )=-\infty $. Also, granting this technicality, the above line
of reasoning gives conditions of unboundedness to the right (one is able to
make the right-hand side of the inequality (4.2.1a) bigger than any given
real by choosing $L$ sufficiently large); in particular, in the same way, it
gives conditions for making the right-hand side bigger than $1$.

Assume the functions are continuous in $L$. An application of the
intermediate-value theorem will then establish the existence of $L_{VT}$
provided there is a value of $L$ for which the right-hand side of (4.2.1a)
is smaller than $1$. There may be no way other than to postulate this, and
it is most sensible to do so for the smallest value $L$ can take, namely $0$.
\end{proof}

\subsubsection{Good-news case}

We proceed similarly in this case, rewriting the characterizing equation
again so as to reflect the relevant conditioning in (4.0.1). The difference
here is that now a reversal of inequalities in the passage from bad- to
good-news decisions requires corresponding new function definitions (below).
These, alongside the term $\mathcal{N}_{VT}(L)$ from (4.2.1b), recast the
existence problem to solving for $L_{VT}$ the equation%
\[
1=\mathcal{N}_{VT}(L_{VT})+\mathcal{GS}_{VT,1}(L_{VT})+\mathcal{GS}%
_{VT,2}(L_{VT}).\leqno{\rm(4.2.2a)}
\]%
Here (with `G for good news') we define:
\[
\mathcal{GS}_{VT,1}(L)=\int_{\mathcal{C}_{VT}}E^{Q_{VT}}[X_{u}\mathbf{1}%
_{\{X_{u}\geq (1+a_{VT,u})L\}}\,|\,\mathcal{F}_{t}\,]\,\frac{\mathrm{d}u}{T-t%
}
\]%
\begin{equation}
\qquad +\,q_{VT}\sum\nolimits_{u\in \mathcal{D}_{VT}}E^{Q_{VT}}[X_{u}\mathbf{%
1}_{\{X_{u}\geq (1+a_{VT,u})L\}}\,|\,\mathcal{F}_{t}\,],\text{ }%
\leqno{\rm(4.2.2b)}  \nonumber
\end{equation}%
\[
\mathcal{GS}_{VT,2}(L)/L=\int_{\mathcal{C}_{VT}}E^{Q_{VT}}[\mathbf{1}%
_{\{X_{u}<(1+a_{VT,u})L\}}\,|\,\mathcal{F}_{t}\,]\,\frac{du}{T-t}
\]%
\[
\qquad +\,q_{VT}\sum\nolimits_{u\in \mathcal{D}_{VT}}E^{Q_{VT}}[\mathbf{1}%
_{\{X_{u}<(1+a_{VT,u})L\}}\,|\,\mathcal{F}_{t}\,].\leqno{\rm(4.2.2c)}
\]%
Their behaviour and consequent relation to the existence of a target value
is again captured by a result analogous to Proposition 4.2.

\bigskip

\begin{proposition}
\quad \textit{\ In good-news events, for a constant $L_{VT}$ to exist for
which equation $(4.2.6a)$ holds, the two conditions \textrm{(i)} and \textrm{%
(ii) below are sufficient}.}

\begin{itemize}
\item[\textrm{(i)}] $\mathcal{GS}_{VT,1}$\textit{\ and $\mathcal{GS}_{VT,2}$
are continuous maps on $[0,\infty )$.}

\item[\textrm{(ii)}] \textit{$1\geq \mathcal{GS}_{VT,1}(0)$.}
\end{itemize}

\bigskip
\end{proposition}

\begin{proof}
Mutatis mutandis, the line of reasoning developed
for Proposition~4.2 now applies. Here, the larger $L$ is, the closer the
indicator functions in $\mathcal{GS}_{VT,2}(L)$ will come to the indicator
function of the entire space; this translates into $\mathcal{GS}_{VT,2}(L)$
becoming similar to some real $\mathcal{GS}_{VT,2}(\infty )$ as $L$ grows
large, and this real is positive. Since $\mathcal{GS}_{VT,1}(L)\geq 0$ for
every $L\geq 0$, the line of reasoning of Section~4.2.1 now establishes
without further conditions the unboundedness in $L$ of the right-hand side
of (4.2.6a).
\end{proof}

\subsubsection{Worked Example in the Geometric Brownian case}

Corresponding to the mark-up decision rules of (2.0.1) there are six
expectations appearing in the formulas of section \S 4.2.1 and \S 4.2.2 that
are needed for an explicit determination of $L_{VT}$. Assume that $X$
follows geometric Brownian motion:%
\[
X_{t+s}=X_{t}\exp \left( (\mu _{VT}-\frac{1}{2}\sigma _{VT}^{2})s+\sigma
_{VT}W_{VT,s}\right) \qquad s\in \lbrack 0,\infty ),
\]%
with $\mu _{VT}\in \mathbb{R}$, $\sigma _{VT}>0,$ and $W_{VT}$ a standard
Brownian motion independent of time $t$ information $\mathcal{F}_{t}$ as in
(3.2.1). For fixed $u=t\!+\!s$ in $[t,T_{1}]$, these six expectations are
provided by standard results on Brownian motion. Corresponding to (4.2.1c)$%
_{\infty }$ one has
\[
E^{Q_{VT}}[X_{u}\,|\,\mathcal{F}_{t}\,]=X_{t}\exp (\mu _{VT}s);%
\leqno{(4.2.3)}
\]%
likewise, corresponding to the pair (4.2.1d), (4.2.6c) and the pair
(4.2.1c), (4.2.6b), taking
\[
\Delta _{s}:=\frac{\log ((1+a_{VT})L/X_{t})-(\mu _{VT}+\frac{1}{2}\sigma
_{VT}^{2})s}{\sigma _{VT}},\leqno{(4.2.4)}
\]%
one has respectively:%
\[
\leqno{(4.2.5a)}\qquad E^{Q_{VT}}[\mathbf{1}_{\{X_{u}\geq
(1+a_{VT})VT\}}\,|\,\mathcal{F}_{t}\,]=\frac{1}{2}\mathrm{Erfc}\left( \Delta
_{s}/\sqrt{2s}\right) ,\,
\]%
\[
\leqno{(4.2.5b)}\qquad E^{Q_{VT}}[\mathbf{1}_{\{X_{u}\leq
(1+a_{VT})VT\}}\,|\,\mathcal{F}_{t}\,]=\frac{1}{2}\mathrm{Erfc}\left(
-\Delta _{s}/\sqrt{2s}\right) ,
\]%
\[
\leqno{(4.2.6a)}E^{Q_{VT}}[X_{u}\mathbf{1}_{\{X_{u}\geq (1+a_{VT})VT\}}\,|\,%
\mathcal{F}_{t}\,]=\frac{1}{2}X_{t}\exp (\mu _{VT}s)\mathrm{Erfc}\left(
\frac{\Delta _{s}-\sigma _{VT}s}{\sqrt{2s}}\right) ,
\]%
\[
\leqno{(4.2.6b)}E^{Q_{VT}}[X_{u}\mathbf{1}_{\{X_{u}\leq (1+a_{VT})VT\}}\,|\,%
\mathcal{F}_{t}\,]=\frac{1}{2}X_{t}\exp (\mu _{VT}s)\mathrm{Erfc}\left( -%
\frac{\Delta _{s}-\sigma _{VT}s}{\sqrt{2s}}\right) .
\]%
Here $\mathrm{Erfc}$ is again the complementary error function, for which
specifically see Appendix equations~(A.7a, b) and (A.8).

For periods of continuous monitoring, integrals of these three expressions
need to be computed over time $s$. This is unproblematic for (4.2.2), where
it reduces to differencing of the right-hand side across the endpoints of
the monitoring period (and division by $\mu _{VT}$). For (4.2.5ab) and
(4.2.6ab) this will lead to expressions in terms of non-standard special
functions: the incomplete Bessel functions, given by integrals of the form $%
\int_{\raise2pt\hbox{$\scriptstyle [t,T]$}}x^{\alpha }\exp
(-(A/x^{2}\!+\!B^{2}x))\,\mathrm{d}x$ for some real constants $A$, $B\geq 0$
and $\alpha $. Series representations can be derived for these integrals;
generically in $\alpha $, the series are in terms of values of the
incomplete gamma function, namely%
\[
\int_{\lbrack t,T]}x^{\alpha }\exp (-(A/x^{2}+B^{2}x))dx=\leqno{(4.2.7)}
\]%
\[
=(B^{2})^{\alpha +1}\sum\nolimits_{m=0}^{\infty }\frac{(-1)^{m}}{m!}%
(B^{2})^{2m}\{\Gamma (-(\alpha +m+1),\frac{B}{T})-\Gamma (-(\alpha +m+1),%
\frac{B}{t})\},
\]%
where the series may be expressed in terms of $\mathrm{Erfc}\,$only for
particular choices of $\alpha $ (integer or half-integer values).

\section{Comparative statics of early disclosure}

Here we address matters on which \cite{GieO} is silent.

\subsection{General performance index $\Sigma $}

\begin{theorem}
\quad \textit{With the modelling assumptions of Section~3.1, the following
assertions hold in the framework of Sections~4.1 and 4.2.}

\begin{itemize}
\item[\textrm{(1)}] \textit{Time-$T$ disclosure becomes the more likely the
smaller are $a^{\ast }$ and $VT$.}

\item[\textrm{(2)}] \textit{\ In good-news situations time-$T$ disclosure is
the more likely the smaller is $r$ and the bigger are $a$ and $E_{T}^{\ast }$%
.}

\item[(3)] \textit{In bad-news situations, time-$T$ disclosure is the more
likely the bigger is $r$ and the smaller are $a$ and $E_{T}^{\ast }$.}
\end{itemize}
\end{theorem}

\bigskip

\begin{proof}
Working in the general process situation of
Sections~4.1 and 4.2, the starting point here are the inequalities holding
at any time $T$ in $(T_{0},T_{1}]$ which trigger a disclosure. In a
good-news event this is%
\begin{equation*}
V_{T}\geq (1+a)e^{-r(T_{1}-T)}E_{T}^{\ast }Q^{\ast }(\Sigma _{T,T_{1}}^{\ast
}\geq (1+a^{\ast })VT|\mathcal{F}_{T}^{\ast });\leqno{(5.1.1)}
\end{equation*}%
and for bad-news this is%
\begin{equation*}
V_{T}\leq (1+a)e^{-r(T_{1}-T)}E_{T}^{\ast }(1-Q^{\ast }(\Sigma
_{T,T_{1}}^{\ast }>(1+a^{\ast })VT|\mathcal{F}_{T}^{\ast })).\leqno{(5.1.2)}
\end{equation*}%
Treating $VT$ as a given (computed by the accounts department), the
likelihood of the validity of these inequalities is determined by the size
of the respective right-hand side; the good-news event (5.1.1) occurring is
the more likely the bigger is the size of the expression on the right; the
bad-news event (5.1.2) becomes the more likely the smaller is the size of
the right. The variables on which the validity of these inequalities depend
are thus: $a$ and $VT$, correspondingly $a^{\ast }$ and $E_{T}^{\ast }$, the
interest-rate $r,$ the `time to maturity' $T_{1}-T$ (time left to the
mandatory disclosure), and the variables beyond these that enter into the
construction of $\Sigma _{T,T_{1}}^{\ast }$; the latter variables include $%
VC $ via $S_{T}^{\ast }$.
\end{proof}

\bigskip

\begin{remark}
\textit{The effects of $S_{T}^{\ast }$ and $T_{1}\!-\!T$ on a time-$T$
disclosure decision depend on the specific form of the law of $\Sigma
_{T,T_{1}}^{\ast }$. }To justify this assertion as supplementary to
Theorem~5.1, we look at the effects of infinitesimal changes in $S_{T}^{\ast
}$ and $T_{1}\!-\!T$ on the (right-hand sides of) (5.1.1) and (5.1.2).
Granted for simplicity the partial differentiability of these probabilities,
we have the following two equations in terms of $\lambda _{T}^{\ast }$, the
law of $\Sigma _{T,T_{1}}^{\ast }$ contingent on time-$T$ information $%
\mathcal{F}_{T}^{\ast }$:%
\[
\partial _{T_{1}-T}[e^{r(T-T_{1})}Q^{\ast }(\Sigma _{T,T_{1}}^{\ast
}>(1+a^{\ast })VT|\mathcal{F}_{T}^{\ast })]\leqno{(5.1.3)}
\]%
\[
=e^{r(T-T_{1})}\int_{(1+a^{\ast })VT}^{\infty }\{\partial _{T_{1}-T}\lambda
_{T}^{\ast }-r\lambda _{T}^{\ast }\}(u)\text{ }\mathrm{d}u,
\]%
\[
\partial _{S_{T}^{\ast }}[e^{r(T-T_{1})}Q^{\ast }(\Sigma _{T,T_{1}}^{\ast
}>(1+a^{\ast })VT|\mathcal{F}_{T}^{\ast })]=e^{r(T-T_{1})}\int_{(1+a^{\ast
})VT}^{\infty }\partial _{S_{T}^{\ast }}\lambda _{T}^{\ast }(u)\text{ }%
\mathrm{d}u.\leqno{(5.1.4)}
\]%
We see from these two equations that the signs of the effects depend on the
exact form of this law, and so need to be determined on a case by case
basis. Suffice it to say that conditions needing to be imposed here in
concrete cases of $\lambda _{T}^{\ast }$ include conditions that entail that
no sign changes occur in the respective integrands on the right-hand sides
of these equations. To determine these signs explicitly requires concrete
choices, at the least for $S^{\ast }$ and for how exactly $S^{\ast }$ enters
into the definition of $\Sigma _{T,T_{1}}^{\ast }$.
\end{remark}

\subsection{Comparative statics for changes in $S_{T}^{\ast }$}

We show that the assumption that the process $S^{\ast }$ follows a Markov
process is sufficient for the determination of the effect of $S_{T}^{\ast }$
on early disclosure. So we work with processes with two properties: firstly
that, for any time $u\geq 0$,
\[
S_{T_{0}+u}^{\ast }=S_{T_{0}}^{\ast }\exp (X_{u}^{\ast }),\quad \hbox{where}%
\quad S_{T_{0}}=VC\,,\leqno{(5.2.1)}
\]%
where $(X_{u}^{\ast })_{u\geq 0}$ is a process independent of time-$T_{0}$
information $\mathcal{F}_{T_{0}}^{\ast }$; secondly that, also for \textit{%
arbitrary} $T\in (T_{0},T_{1}],$
\[
S_{T+u}^{\ast }=S_{T}^{\ast }\exp (X_{u}^{\ast })\,,\quad u\in \lbrack
0,\infty ),\leqno{(5.2.2)}
\]%
where, by abuse of language, $(X_{u}^{\ast })_{u\geq 0}$ (or a suitable
version of the process in (4.3.5) denoted by the same same symbol) is a
process independent of time-$T$ information $\mathcal{F}_{T}^{\ast }$.

\bigskip

\begin{theorem}
\quad \textit{We have $\partial _{S_{T}^{\ast }}Q^{\ast }\big(\Sigma
_{T,T_{1}}^{\ast }>(1\!+\!a^{\ast })VT\,|\,\mathcal{F}_{T}^{\ast }\big)>0$
in $(4.3.4)$ under the additional Markovian conditions $(5.2.1)$ and $%
(5.2.2) $ on $S^{\ast }$; so in both good-news and bad-news situations a
time-$T$ disclosure becomes the more likely the larger is $S_{T}^{\ast }$.
This conclusion holds more generally for all constructions of $\Sigma
_{T,T_{1}}^{\ast }$ that preserve scaling }(\textit{in the sense that we
have $\Sigma _{T,T_{1}}^{\ast }(S^{\ast })=S_{T}^{\ast }\Sigma
_{T,T_{1}}^{\ast }(\exp (X^{\ast }))$ }).
\end{theorem}

\bigskip

\noindent \textbf{Proof. }The assumption (5.2.1) and (5.2.2) imply that $%
S^{\ast }$ is the product of a scaling factor $S_{T}^{\ast }$ and a
`standardized version' of $S^{\ast }$. This yields a representation for $%
\Sigma _{T,T_{1}}^{\ast }(S^{\ast })$ as a product of two positive factors,
the scaling factor $S_{T}^{\ast }$ and the random variable that results from
the application of the respective construction for $\Sigma _{T,T_{1}}^{\ast
} $ to the standardized version of $S^{\ast }$.

In (5.1.1) and (5.1.2) the partials w.r.t. $S_{T}^{\ast }$ of the
probability factors will thus be positive so long as we have independence of
the normalized versions of $S^{\ast }$ from time $T$-information. \hfill $%
\square $

\subsection{Early disclosure for Geometric Brownian performance indices}

We consider here the other parameter mentioned in Remark~5.2 namely $%
T_{1}\!-\!T$, i.e. the time left to the next mandatory reporting date. At
first sight, one might expect a proposition asserting that the shorter is
this time, the less likely are decisions made for an early disclosure.
However, on reflection, such decisions may well depend on the actual
evolution of the market capitalization of the firm, and market forces may
lead to changes in the size of this capitalization forcing early disclosure
also at dates comparatively close to the mandatory date. Therefore, a
discussion of the effects of $T_{1}\!-\!T$ needs to be incorporated in a
model framework that includes $S^{\ast }$.

Here we adopt a standard Black-Scholes modelling for $S^{\ast }$, and
therefore consider now the process $X^{\ast }$ in Section~5.2 as following
scaled Brownian motion with drift:
\[
X_{u}^{\ast }=X^{\ast }(\mu ^{\ast },\sigma ^{\ast })_{u}=\mu ^{\ast
}u+\sigma ^{\ast }W_{u}^{\ast }\,,\quad u\in \lbrack 0,\infty ),\leqno{\rm
(5.3.1)}
\]%
where $W^{\ast }$ is a $Q^{\ast }$-Brownian motion independent of time-$T$
information started with value $0$ at time $0$, with parameters $\sigma
^{\ast }>0$ and $\mu ^{\ast }=r-\Delta ^{\ast }-(1/2)(\sigma ^{\ast
})^{2}\in ~\mathbb{R}$. Recall the former parameter is a measure of market
volatility, while $r\!-\!\Delta ^{\ast }$ is the excess of the (riskless)
short rate, $r$, over the dividend rate of the firm $\Delta ^{\ast }$
\textit{as seen by the markets}; in the present context this difference
should be viewed as an `appreciation rate' for investments in $F$ (again as
seen by the markets).

The four effects on $V^{\ast }$ to consider now are those induced by changes
in $T_{1}\!-\!T$ and also in $\sigma ^{\ast }$, $r\!-\!\Delta ^{\ast }$, and
$r$. These four will depend on which of good-news or bad-news situations
occurs; they enter via the market proxies, and hence even a qualitative
picture will depend on the concrete form of $\Sigma _{T,T_{1}}^{\ast }$. We
focus on modelling $\Sigma _{T,T_{1}}^{\ast }$ as the running minimum or
maximum of $S^{\ast },$ as in equations (3.1.5a) and (3.1.5c), specifically
in the good-news situations, so that by equation (5.1.1) we must consider
the inequalities
\[
V_{T}\geq V_{\bullet }^{\ast }\quad \bullet \in \{\max ,\min \},%
\leqno{(5.3.2)}
\]%
with
\[
V_{\text{max}}^{\ast }:=(1+a)E_{T}^{\ast }\exp (-r(T_{1}-T))Q_{\text{max}%
}^{\ast },\leqno{\rm (5.3.3a)}
\]%
\[
Q_{\text{max}}^{\ast }:=Q^{\ast }(\max_{u\in \lbrack 0,T_{1}-T]}\{X^{\ast
}(\mu ^{\ast },\sigma ^{\ast })\}>A_{T}^{\ast }),\leqno{\rm (5.3.3b)}
\]%
\[
V_{\text{min}}^{\ast }:=(1+a)E_{T}^{\ast }\exp (-r(T_{1}-T))Q_{\text{min}%
}^{\ast },\leqno{\rm (5.3.4a)}
\]%
\[
Q_{\text{min}}^{\ast }:=Q^{\ast }(\min_{u\in \lbrack 0,T_{1}-T]}\{X^{\ast
}(\mu ^{\ast },\sigma ^{\ast })\}>A_{T}^{\ast }),\leqno{\rm (5.3.4b)}
\]%
where we set%
\[
A_{T}^{\ast }=\log ((1+a^{\ast })VT/S_{T}^{\ast }).\leqno{\rm (5.3.5)}
\]

\subsubsection{Explicit results for geometric Brownian performance indices}

We give a paradigm discussion of the effects of $T_{1}\!-\!T$ on the
likelihood of disclosure decisions, when for the good-news case of (5.3.2)
the running maximum is used in the construction of market proxies according
to (5.3.2a,b). These effects turn out to depend on the sign of the mark-up
parameter $A_{T}^{\ast }$ as follows.

\bigskip

\begin{theorem}
\quad \textit{In the framework of Section~5.3, assume a situation of time-$T$
non-disclosure of good news. Using the running maximum of $S^{\ast }$ in the
construction of the market proxy $V^{\ast }$, the following two assertions
are equivalent.}

\begin{itemize}
\item[\textrm{(1)}] \textit{\ Early disclosure in $[T,T_{1}]$ is the more
likely the nearer is $T$ to $T_{1}$. }

\item[\textrm{(2)}] \textit{\ $\mathrm{sign}(\partial _{T_{1}-T}V_{\max
}^{\ast })<0$.}
\end{itemize}

\textit{\noindent Here the partial derivative in $(2)$ depends on the sign
of $A_{T}$; for $A_{T}^{\ast }\leq 0$ this is
\[
\partial _{T_{1}-T}V_{\max }^{\ast }=-rV_{\max }^{\ast },
\]%
while for $A_{T}^{\ast }\geq 0$ this is}%
\begin{eqnarray*}
\partial _{T_{1}-T}V_{\max }^{\ast } &=&-rV_{\max }^{\ast }+(1+a)E_{T}^{\ast
}\exp (-r(T_{1}\!-\!T))\,\frac{A_{T}^{\ast }\text{\textit{\ }}\exp (-\eta
^{2})}{\sigma ^{\ast }\sqrt{{2\pi (T_{1}-T)^{3}\,}}}, \\
\text{where }\eta &=&(A_{T}^{\ast }-(T_{1}\!-\!T)\mu ^{\ast })/\sigma ^{\ast
}\sqrt{{2(T_{1}-T)\,}}.
\end{eqnarray*}
\end{theorem}

To indicate the typical line of reasoning for results like this, start from
(5.3.2), observing that (in these good-news situations) disclosure decisions
at some fixed point in time $T$ are the more likely the larger $V_{T}^{\ast
} $ is. The effects of some parameter on the likelihood of early disclosure
thus translate into the determination of the corresponding partial
derivative of $V_{T}^{\ast }$, and so a determination of their qualitative
effect reduces to a determination of the sign of these partials. Early
disclosure thus becomes more likely the larger the relevant parameter,
provided the corresponding partial of $V_{T}^{\ast }$ is positive, and vice
versa. The point of our choice of a geometric Brownian framework is that
explicit formulas for the probability factors $Q_{\bullet }^{\ast }$ are
available as standard results in Brownian motion; these are reviewed in
Appendix A below, with equations (A7a,b) pertinent for the present case of
running-maximum performance parameters. Establishing comparatice statics
results therefore reduces to straightforward partial differentiation of
explicitly given functions, albeit of some complexity. Theorem~5.3 collects
the results when the relevant parameter there is the time left to the next
mandatory date.

\bigskip

\begin{remark}
Proceeding along the same lines in the same situation, one obtains results
similar to those of Theorem~5.3 concerning the effects of the volatility $%
\sigma ^{\ast }$, whereas the effects of $r\!-\!\Delta ^{\ast }$ and $r$ are
unequivocally unidirectional (with the signs of the pertinent partials being
equal to minus that of $E_{T}^{\ast }$). Provided $E_{T}^{\ast }>0$, early
disclosure within $(T,T_{1})$ is the more likely the smaller are $%
r\!-\!\Delta ^{\ast }$ and $r$.
\end{remark}

\subsubsection{Explicit results for running-min}

Here we note only that if the market proxies are instead constructed using
the running minimum of $S^{\ast }$ analogues of Theorem~5.3 and Remark 5.5
again hold and preserve all the conclusions above except for a sign reversal
in $A_{T}^{\ast }$. This shows how derivation of the effects of $T_{1}\!-\!T$
on early disclosure requires the specifics of a given model.

\section{Managerial Implications}

This paper's approach to asset pricing allows the development of a richer
appreciation of how voluntary disclosure by firms can affect firm asset
valuation in equilibrium. Existing research has typically modelled voluntary
disclosure as the choice by firms to make additional voluntary (content)
disclosures to the market at fixed time points. As such this literature does
not consider the possibility that firms may choose not only \textit{what }to
disclose voluntarily but also \textit{when} to disclose. Thus voluntary
disclosure has at least two dimensions:\textit{\ content} and \textit{timing}%
. As existing models typically do not consider the latter dimension, they
are not truly dynamic, and hence do not provide the necessary building
blocks to develop a realistic empirical model of (`two dimensional')
voluntary disclosure. Here we have explicitly modelled the joint
content-timing interaction, so enabling more realistic formal modelling of
problems faced by managers of firms: when private news is uncertain, how
good does that private news have to be before it is in the interests of the
firm to issue a voluntary disclosure. The other side of this coin is what
materiality standard needs to be followed in managing the voluntary
disclosure process. The comparative statics derived in the preceding section
permit an understanding of how changes in parameter values may explain
differences in equilibrium behaviour between firms -- some voluntarily
releasing additional information early, others not. This meets the challenge
of modelling equilibrium asset-pricing with endogenously determined
voluntary disclosures, wherein both the content and the timing of
disclosures are rationally chosen, making delay or early release of
information in capital markets an equilibrium outcome.

\section{Complements}

We close with some observations about the potential of the approach above
especially with regard to variations on the themes presented and
generalizations away from the Brownian framework followed above.

\subsection{Mechanisms}

Implicit in our development of a markets-based general modelling framework
was the need to pick apart the `who does what and how' into `building
bricks', and with these to build a variety of models. We implicitly
identified five such bricks, which in fact are best considered as
mechanisms, to borrow a phrase from economic theory. These are made explicit
here so as to stress both the sensitivity of a model to its assumptions and
its adaptability to alternative contexts.

\bigskip

\noindent \textit{Mechanism (i). Evolution rules}.\quad The perspective
adopted above is rather like that of a scientist designing experiments and
subsequently observing outcomes and evolution. Ingredients thus include
design dynamics, start time and observation times. Thus mechanisms~(i)
amounts to formal rules for encoding these three aspects. Real-life features
mapped via such `experiments' include informational interplay between
economic agents and firm-to-market communications. See the summary in
Section 7.5 for an explicit illustration of how this can be further
developed.

\bigskip

\noindent \textit{Mechanism (ii). Decision-rule strategies}.\quad The task
here is to provide rules for triggering `events' (typically, public
disclosure of privileged information), and the idea is that these be the
consequence of some `rule', i.e. functional relation, applied to some
observation variables. At a technical level, this mechanism thus amounts to
the selection of functional relations subject to the specification of
observation variables. The mechanisms we adopted, starting from Section~3,
are motivated by the provision of approximations to \textit{%
equilibrium-induced decision rules} as derived in [GieO]. There they
correspond to `value-enhancement' disclosures when observations are
`sufficiently high'. Whilst outside the scope of the present paper, the
argument there may be dualized to correspond to equilibrium-induced
`value-erosion' \textit{alerts} when observations are `sufficiently low',
with a resulting notion of \textit{endogeneous} `materiality thresholds'.
Such an understanding leads to the following

\bigskip

\begin{conjecture}
\quad \textit{First-order approximation of equilibrium decision-rules based
on the notion of materiality yields decision-rules using functions taking
the form }%
\[
\mathit{h_{\varepsilon ,a}(t,x,y)=(x\!-\!(1\!+\!a_{t})y)\varepsilon _{t}}
\]
\textit{for some families of signs $\varepsilon =(\varepsilon_{t})_{t\in
\lbrack T_{0},T_{1}]}$ and mark-ups $a=(a_{t})_{t\in \lbrack T_{0},T_{1}]}$,
and vice versa.}
\end{conjecture}

\bigskip

\noindent \textit{Mechanism (iii). Observation processes $V$}.\quad A
further fundamental notion is informational assymetry, the task being to
construct a `variable' (or perhaps a vector) with two properties. Firstly,
it is capable of observation over time and is observed over time by the
informationally privileged agents of whatever model is to be constructed
(denoted by the symbol $F$, typifing firms); secondly, the variable is at
best partially observable by the remaining agents (denoted by the symbol $M$%
).

As to observation variables, we focus on a \textit{portfolio view}.
Continuing to think of $F$ as a firm for a moment, $F$ will not in general
observe just a single source of information to set a target, but a portfolio
of these, say $(X_{1},\ldots ,X_{n},\ldots ,X_{N})$. Formally, the \textit{%
chosen observation process} $X$ will be a function of the $X_{n}$; simple,
but typical, functions are linear or multiplicative forms in the $X_{n}$ as
given respectively by
\[
X=\sum\nolimits_{n=1}^{N}a_{n}X_{n},\text{ or }X=\prod%
\nolimits_{n=1}^{N}X_{n}^{a_{n}},
\]%
with suitable real weights $a_{n}\in \mathbb{R}$.

\bigskip

\textit{Examples of two factor portfolio observation variables}.\quad In the
paper, $V$ was interpreted as representing the value of the firm $F$, i.e. a
process internally observed by $F$. It is natural to complement it by a
process that encodes the external view of the firm's value, such as provided
by the the firm's market capitalization, $S^{\ast }$. Specializing to
portfolios of additive type, the associated observation variable may take
the form
\[
X=aV+bS^{\ast },
\]%
for some $a,b\in \mathbb{R}$. General structure of $X$ apart, the modelling
of $V$ and $S^{\ast }$ is far from straightforward, and Sections~7.2 to 7.4
below offer an amendment to the simplified treatment given in the main body
of the paper.

\bigskip

\noindent \textit{Mechanism (iv). Observation process proxies $V^{\ast }$
from state observer systems}.\quad The task here is to enable specifically
the `informationally under-privileged' agents in the models to approximate $V
$. To paraphrase a key idea in the paper: here $S^{\ast }$ is seen as a
public proxy for $V$ in creating an estimate $V_{t}^{\ast }$ of $V_{t}$; for
tractability we made specific approximation assumptions.

In so doing, we borrowed an idea from the control theory of an engineering
plant, where one way to deal with imperfect information about the plant is
to build a laboratory version (a model) of the plant with accessible
full-information of its state at any time (known as a `state observer'
system [Rus, Ch. 3], [Son, Ch. 7] -- in reality a `state estimator');
state-correcting signals are sent to this model, using plant-based,
imperfect, or partial observations, with which to guide the `observer
system' (model) into greater agreement with the plant.

Unlike in the engineering context, inclusion into a market-based model of an
observer-style system implies changes to the strategic behaviour of the firm
in its decisions to hide certain observations of its state. Indeed, here
\textit{each} side (each of the agents, $F$ or $M$) enriches its algorithmic
opportunities. In this context, our version of an `observer system' responds
to strategic behaviour, so is far richer. To emphasize this difference when,
for example, $F$ was a firm, the observer-style system was termed a \textit{%
proxy-firm}.

\bigskip

\noindent \textit{Mechanism (v). Target value processes $VT$}.\quad Under
this heading, the task is to provide techniques for forecasting values of
observation variables (and, significantly, of their proxies), bearing in
mind that such items are contingent on market developments as well as on
restrictions arising from (production) technologies. In Section~4 we marry
accounting analysis with the analysis of Bermudean options: mechanism~(v)
involved observation levels $L_{t}$ for $V_{t}$ qua strike prices at which
the respective agent's decision rule is indifferent between disclosing or
suppressing private observation of $V_{t}$, were $V_{t}$ to take the value $%
L_{t}$ (as noted in the introduction). This construction borrows from the
stylized model of [GieO], where these levels describe the market's current
view of the value during periods of silence and so provide the basis of
current guidance on its earnings target.

\subsection{Modelling firm-value processes: uncertainty structure of profits}

As to the observation mechanism, one may `drill down' to the basic structure
of profits and address the uncertainty effects created by reporting lags.
The starting point is a formalization of accounting practice: $V_{u}$ the
firm $F$'s time-$u$ value is the accrual of an instantaneous variable $\pi
_{w}$ over the period $[T_{0},u]$ added to an initially given value $VC$ of
time $T_{0}$, so that
\[
V_{u}=VC+\int_{[T_{0},u]}\pi _{w}\,\mathrm{d}w,\quad u\in \lbrack
T_{0},T_{1}],\leqno{(7.2.1)}
\]%
implicitly assuming $w\mapsto \pi _{w}$ to be summable over $[T_{0},T_{1}]$.

\bigskip

The simplest interpretation of $V_{u}$ is incrementing $VC$ by the firm's
\textit{actual} profit flow rate $\pi _{w}$ -- as it actually arises at each
time moment $w$ between time-$T_{0}$ and time-$u$. Alternatively, to allow
for the possibility of delays in reporting profits (due, say, to reporting
delays of costs, as below), we can re-interpret this as the \textit{%
recognized profit} flow rate -- namely, the value posted in some official
ledger.

\bigskip

To introduce \textit{reporting lags} into the model, fix $\Lambda \geq 0$
and then at each time $u$, assume the flow $\pi _{w}$ is certain only for
`distant' times $w$, namely times earlier than $u\!-\!\Lambda $, but for
times $w$ nearer to $u$, i.e. in $(u\!-\!\Lambda ,u]$, $\pi _{w}$ is
uncertain. A further refinement then occurs in the decomposition (7.2.1)
created by a \textit{deterministic\/} part $\Delta ^{\mathrm{ns}}{V}_{u}$ ,
certain at time-$u$ (with `ns' for non-stochastic), and a part $\Delta ^{%
\mathrm{s}}{V}_{u}$ \textit{uncertain\/} at time-$u$:
\[
V_{u}=VC+\Delta ^{\text{ns}}V_{u}+\Delta ^{s}V_{u},\qquad u\in \lbrack
T_{0},T_{1}],\leqno{\rm(7.2.2a)}
\]%
where
\[
\Delta ^{\text{ns}}V_{u}=\int_{[T_{0},\max \{T_{0},u-\Lambda \}]}\pi _{w}\,%
\mathrm{d}w,\leqno{\rm(7.2.2b)}
\]%
\[
\Delta ^{\text{s}}V_{u}=\int_{[\max \{T_{0},u-\Lambda \},u]}\pi _{w}\,%
\mathrm{d}w.\,\leqno{\rm(7.2.2c)}
\]%
We view ${V}_{u}^{\mathrm{ns}}=VC\!+\!\Delta ^{\mathrm{ns}}{V}_{u}$ as an
\textit{accounting equality}, namely, as data held, or stored, by the firm $%
F $, and $\Delta ^{\mathrm{s}}{V}_{u}$ as a variable that needs to be
modelled.

Two obvious questions arise: first, how does (a manager) $F$ respond to such
operational uncertainty. For example, will there be a period in which $F$ is
waiting for the time-$T$ accounting information to be corroborated and to be
verified as reliable (up to a level of security deemed appropriate for the
decision-making), and how does $F$ then respond to the evolution of market
sentiment during such periods of waiting? Second, is there a correlation
between market sentiment $V^{\ast }$ and the degree of operational
effectiveness of the firm's accounting department?

\subsection{Modelling firm-value processes: Cobb-Douglas examples}

As a second complement to our discussion of Mechanisms~(iii), we provide
examples for modelling firm-value observation processes $V$ concretely.

\subsubsection{(i) \textit{Deterministic}}

Here the construction of $V$ needs to be linked to the standard functional
forms preferred by the theory of the firm in Economics and Econometrics. We
consider Cobb-Douglas technologies, and indicate how to model profits
derived from a Cobb-Douglas production function ([Var, Ch. 1], [Rom, Ch.~2])
corresponding to a \textit{single} output from \textit{two} input factors
(such as capital and labour) with respective parameters $a,b\geq 0$
satisfying $a\!+\!b<1.$ This yields profits as a function of input prices $%
w=(w_{1},w_{2})$ and output selling-price $p$ in the form
\[
\leqno{(7.3.1)}\qquad \pi ^{CD}(p,w)=
\]%
\[
=p^{1/(1-(a+b))}\left\{ \Big({\frac{a\!+\!b}{\kappa \cdot c(w)}}\Big)%
^{(a+b)/(1-(a+b))}-\kappa \cdot c(w)\Big({\frac{a\!+\!b}{\kappa \cdot c(w)}}%
\Big)^{1/(1-(a+b))}\right\} ,
\]%
for a cost function $c(w)=(w_{1}^{a}w_{2}^{b})^{1/(a+b)}$, where
\[
\kappa :=((a/b)^{a/(a+b)}+(a/b)^{-(a/(a+b)})/A^{1/(a+b)},\text{ with }%
A=(1\!-\!a)^{a-1}/a^{a}.
\]

\subsubsection{(ii) \textit{Stochastic Cobb-Douglas profits}}

To take into account uncertainties in the profit function, assume that
uncertainty in the output and input price is given by positive stochastic
processes on some stochastic basis, say $\mathcal{X}(Q)=(\Omega ,{\mathcal{F}%
},\mathbf{F}=({\mathcal{F}}_{u})_{u\in \lbrack T_{0},T_{1}]},Q)$. As regards
output, passing to logarithmic prices and so to an exponential price process
$\exp Y$, (7.2.1) yields profits of the form
\[
\pi _{w}=\alpha \exp (Y_{w})^{\beta }=\exp (\alpha (\log (\alpha )+\beta
Y_{w})\,,\quad w\in \lbrack T_{0},T_{1}]\,,\leqno{(7.3.2)}
\]%
for some fixed $\alpha >0$, $\beta \in \mathbb{R}$, and some fixed
stochastic process $Y$ on $\mathcal{X}(Q)$. In turn this gives the uncertain
part of the time-$u$ value of $V$ the representation:
\[
\Delta ^{\mathrm{s}}{V}_{u}=\alpha \int_{\lbrack \max \{T_{0},u-\Lambda
\},u]}\exp (\beta Y_{w})\,\mathrm{d}w,\quad u\in \lbrack T_{0},T_{1}]\,.%
\leqno{(7.3.3)}
\]

Treating input (factor) prices $w$ in similar vein preserves this general
form for the uncertain parts of $V$ in (7.2.2c). In any of these
representations, a notable choice for $Y$ is Brownian motion on $\mathcal{X}%
(Q)$, and this provides an explicit illustration of how the modelling above
turns $V_{u}$ itself into a random variable, given the time-$u$ information
(see Section~7.4 for scalable processes, which we suggest as candidate
modelling mechanisms).

As mentioned, accruals in (7.2.1) can be modelled in (at least) two
conceptually different ways, depending on whether instantaneous profits $\pi
_{s}$ or their accumulated value is taken as a primary variable. The matter
of choice is not just a conceptual but also a practical one, even assuming
the classical Cobb-Douglas two-factor production technology above.
Neoclassical economic theory asserts (cf. \S 3.3, [Rom] and [Var]) the
profits of the firm will then be a function of factor prices $w_{1}$, $w_{2}$
and of the commodity price $p$ taking the form
\[
\leqno{\rm(7.3.4a)}\qquad \pi
_{CD}(p,w_{1},w_{2})=w_{1}^{a/(a+b-1)}w_{2}^{b/(a+b-1)}p^{(a+b-1)}\quad
\cdot
\]%
\[
\qquad \cdot \quad \left\{ \Big({\frac{a\!+\!b}{\kappa }}\Big)%
^{(a+b)/(1-(a+b))}-\kappa \Big({\frac{a\!+\!b}{\kappa }}\Big)%
^{1/(1-(a+b))}\right\} ,
\]%
with constants $a$, $b>0$ such that $a\!+\!b<1$ for
\[
\kappa
:=((a/b)^{b/(a+b)}+(a/b)^{-a/(a+b)})/((1\!-\!a)^{a-1}/a^{a})^{1/(a+b)}\,.%
\leqno{\rm(7.3.4b)}
\]%
Assume that non-deterministic profits are the result of fluctuations in any
of these prices, and, for simplicity, staying within the Brownian framework,
assume the fluctuations follow geometric Brownian motion. The resulting
dynamic for $\pi _{\mathrm{CD}}$ is of the form
\[
\pi _{\mathrm{CD},T+u}=\pi _{\mathrm{CD},T}\exp (\mu _{\pi }u+\sigma _{\pi
}W_{\pi ,u})\,,\quad u\in \lbrack 0,\infty )\,,\leqno{(7.3.5)}
\]%
with driver a Brownian motion $W_{\pi }$ independent of time-$T$ information
$\mathcal{F}_{FF,T}$ and with real constants $\sigma _{\pi }>0$ and $\mu
_{\pi }$. Return now to the choice of explicit modelling variants; according
to the choice of accumulated profits or instantaneous profits, one has
respectively
\[
\pi _{s}=\pi _{CD,T+s},\quad s\in \lbrack 0,\infty )\,,\leqno{\rm(7.3.6a)}
\]%
\[
\int_{\lbrack T,T+s]}\pi _{w}\text{ }\mathrm{d}w=\pi _{CD,T+s},\quad s\in
\lbrack 0,\infty )\,.\leqno{\rm(7.3.6b)}
\]
This last requires for $V_{T+u}$ the integral of geometric Brownian motion,
not covered by the context of Section~4. For the first, the results of
Section~4.2 do apply, and provide the forecast target $VT$.

\subsection{Scalable processes}

As a third complement to Section~7.1, we suggest the use \textit{scalable
processes} for modelling with mechanisms. These processes will be patterned
after the exponentials of strong Markov processes $S^{\ast }$, which satisfy
two equations. Firstly, with $T_{0}$ fixed, for any time $u\geq 0$,
\[
S_{T_{0}+u}^{\ast }=S_{T_{0}}^{\ast }\exp (X_{u}^{\ast }),\quad \hbox{where}%
\quad S_{T_{0}}=VC\,,\ \leqno{(7.4.1)}
\]%
where $(X_{u}^{\ast })_{u\geq 0}$ is independent of time-$T_{0}$ information
$\mathcal{F}_{T_{0}}^{\ast }$. Secondly, for arbitrary real $T$ in $%
(T_{0},T_{1}]$, the representation
\[
S_{T+u}^{\ast }=S_{T}^{\ast }\exp (X_{u}^{\ast })\,,\quad u\in \lbrack
0,\infty ),\leqno{(7.4.2)}
\]%
where, by abuse of language, $(X_{u}^{\ast })_{u\geq 0}$ (or a suitable
version of the process in (7.4.1) denoted by the same same symbol) is a
process independent of time-$T$ information $\mathcal{F}_{T}^{\ast }$. We
now relax the second condition and define processes $S^{\ast }$ to be
\textit{scalable\/} if they are RCLL and satisfy (7.4.1) and (7.4.2), except
that now (7.4.2) need not necessarily hold for \textit{all} $T$ in $%
(T_{0},T_{1}]$ and instead is to hold necessarily for all $T$ in some
prespecified finite subset $\mathcal{T}$ of $(T_{0},T_{1}]$. Here we
primarily think of $\mathcal{T}$ as containing the endpoints of benchmark
observation schemes along the lines of equation (4.1.1). Extending this
notion of scalable process to allow the sets $\mathcal{T}$ to have at most
countably many stopping times should not, however, pose problems.

\subsection{Modelling with mechanisms: a summary vista}

A characteristic feature of the mechanisms identified in Section 7.1 is that
they identify the economic agents solely in terms of how they act. In
respectively Sections 4 and 3, as it happens, the agents $F$ posited by the
mechanisms in disclosure situations are indeed interpreted as acting like
the manager of the firm, and agents $M$ as acting like representative market
participants. For the wider guidance theme, however, specific \textit{market}
participants will also `act out' the role of agent $F$ within some of the
mechanisms. An outline follows.

To tell our guidance story we need to single out a distinguished group of
people from among the market participants $M$, whom we shall call \textit{%
analysts}. The typical representative member of this group being denoted by $%
A$, we continue to denote representative market participants as agents $M$
(as in Section 3).

The guidance theme then starts at time $t$ with the announcement by manager $%
F$ of the current accounting numbers of the firm: $VC$ and its target $VT$
for the next official reporting date $T$. These numbers are processed by $M$
as in the disclosure theme, while now $A$ is also assumed to make its own
computations. For these computations $A$ will be asumed to use mechanism (v)
of Section 7.1 (as though in the role of agent $F$), and make its own
computions of the time-$T$ target, possibly based on a re-estimation of $V$,
$V^{\ast }$, $S$ and $S^{\ast }$; call the result $VT_{A}$, and assume that $%
A$ will announce this number to $M$ and $F$ at a time $t+\Delta _{A}$ within
$[t,T]$. This announcement will possibly induce a re-calibration of the
price processes $S$ and $S^{\ast }$; the analyst's target $VT_{A}$ will now
be added as a new variable in the decision making of the manager $F$. Apart
from a possible consequent re-calibration of $V^{\ast }$, the essential
difference from the set-up of Section 3 is that now manager $F$ is assumed
to use decision rules in four variables, say $h(T,V_{T},V_{T^{\ast }},VT_{A})
$. For present outline purposes, we will not make this 4-variable rule
explicit, leaving the details to be established elsewhere. Now running the
disclosure-mechanisms of Sections 3, 4, and 5 results, mutatis mutandis, in
either an early disclosure at some time $\tau <T$, or a regular one at time $%
T$. In both cases an announcement is made by manager $F$ of new current
numbers $VC$ and targets $VT$ for the next reporting date; these numbers
will be announced simultaneously to the analysts $A$ and to the market (as
represented by agent $M$) and the entire activity starts all over again. The
details are intended to be established elsewhere.

\section*{Appendix: Reductions}

This section collects the simplifications arising for general mark-up
decision rules in Brownian contexts. We work with a fixed probability space $%
(\Omega ,\mathcal{F},P)$ which is equipped with a filtration $\mathbf{F}=(%
\mathcal{F}_{u})_{u\geq 0}$ such that the resulting stochastic basis $%
\mathcal{X}(P)=(\Omega ,\mathcal{F},\mathbf{F},P)$ satisfies the usual
conditions (see for example {\cite[Def.~1.3, p.~2]{JacS}}). The mark-up
decision rules are assumed, as above, in the form $h_{\varepsilon
,a}(x,y)=\varepsilon (x\!-\!(1\!+\!a)y)$, for fixed parameters $\varepsilon
=\pm 1$ and $a>-1$.

\bigskip

\subsection*{A.1\quad Bad-news to good-news reductions:}

For $\Sigma $ a random variable on $\Omega ,$ measurable with respect to $%
\mathcal{F}_{t}$, and fixed $t>0$,%
\[
E^{P}[\mathbf{1}_{\{h_{+1,a}(\Sigma ,VT)\geq 0\}}\,|\,\mathcal{F}%
_{t}]=P(\{\Sigma \geq (1+a)VT\}|\,\mathcal{F}_{t}),\leqno{\rm (A.0a)}
\]%
\[
E^{P}[\mathbf{1}_{\{h_{-1,a}(\Sigma ,VT)\geq 0\}}\,|\,\mathcal{F}%
_{t}]=1-P(\{\Sigma \geq (1+a)VT\}|\,\mathcal{F}_{t}).\leqno{\rm (A.0b)}
\]%
These two yield a reduction of bad-news to good-news disclosures via
\[
E^{P}[\mathbf{1}_{\{h_{-1,a}(\Sigma ,VT)\geq 0\}}\,|\,\mathcal{F}%
_{t}]=1-E^{P}[\mathbf{1}_{\{h_{+1,a}(\Sigma ,VT)\geq 0\}}\,|\,\mathcal{F}%
_{t})]\,,\leqno{\rm (A.1)}
\]%
granted absence of point-masses in $\Sigma $ over $(1\!+\!a)VT$, given the
continuous processes in play here. Note the qualitative consequence that
factors influencing the relevant probabilities will have opposite effects on
good-news and bad-news events; an increase in a factor that leads to an
increase of (A.0a) will decrease (A.1) and vice versa.

\bigskip

\subsection*{A.2\quad Running-minimum to running-maximum reductions:}

\quad We collect here the further reductions needed for good-news events
when $S^{\ast }$ is a geometric Brownian motion. With $\Sigma ^{\ast
}=\Sigma _{t,T_{1}}^{\ast }$ the running maximum or the running minimum of $%
S^{\ast }$ on some fixed time interval $[t,T_{1}]$, let $W^{\ast }$ be an $(%
\mathbf{F},P)$-Brownian motion on $\mathcal{X}(P)$ started at $0$ at time $0$%
; for $\sigma ^{\ast }>0$ take
\[
S_{u+t}^{\ast }=S_{t}^{\ast }\exp (\mu ^{\ast }u\!+\!\sigma ^{\ast
}W_{u}^{\ast }),\quad u\in \lbrack 0,\infty ),\leqno{\rm (A.2)}
\]%
with $\mu ^{\ast }=r\!-\!\delta \!-{\frac{1}{2}}(\sigma ^{\ast })^{2}$ for $%
r\!,\delta \in \mathbb{R}$; appealing to the strong Markov property of
Brownian motion, assume also $W^{\ast }$ to be independent of $\mathcal{F}%
_{t}$, and express the events in terms of $W^{\ast }$ as follows:%
\begin{eqnarray*}
\text{(A.3)\qquad }E^{P}[\mathbf{1}_{\{h_{+1,a}(\max \{S_{w}^{\ast }|w\in
\lbrack t,T]\})\geq 0\}}\,|\,\mathcal{F}_{t}] &=&P(\max_{u\in \lbrack
0,T-t]}\{\mu ^{\ast }u+\sigma ^{\ast }W_{u}^{\ast }\}\geq A_{t}^{\ast }), \\
\text{(A.4)\qquad }E^{P}[\mathbf{1}_{\{h_{-1,a}(\min \{S_{w}^{\ast }|w\in
\lbrack t,T]\})\geq 0\}}\,|\,\mathcal{F}_{t}] &=&P(\min_{u\in \lbrack
0,T-t]}\{\mu ^{\ast }u+\sigma ^{\ast }W_{u}^{\ast }\}\leq A_{t}^{\ast }),
\end{eqnarray*}%
with%
\[
A_{T}^{\ast }:=\log ((1+a^{\ast })VT/S_{t}^{\ast }).\leqno{\rm (A5)}
\]%
Setting $W^{\ast \ast }:=-W^{\ast }$ note that%
\[
P(\min_{u\in \lbrack 0,T-t]}\{\mu ^{\ast }u+\sigma ^{\ast }W_{u}^{\ast
}\}\leq A_{t}^{\ast })=P(\max_{u\in \lbrack 0,T-t]}\{-\mu ^{\ast }u+\sigma
^{\ast }W_{u}^{\ast \ast }\}\geq -A_{t}^{\ast }).\leqno{\rm (A6)}
\]%
It is special to the Brownian context that (A.6) provides a reduction of the
running-minimum event in (A.4) to a running-maximum event in (A.3), since,
if $W^{\ast }$ is Brownian, then so is its negative $W^{\ast \ast }$
(above). An explicit determination of the expectation (A.3) can be had from
the explicit law for the running-maximum of Brownian motion; see e.g. {\cite[%
(30)~Corollary, p.~25]{Fre}}. This relies, in this Brownian context, on the
running-maximum always being positive on time intervals of positive length;%
\textbf{\ }indeed, this follows from the fact that the running-maximum of
the process is zero on non-positive time arguments. This is not directly of
use here, since the drift $\mu ^{\ast }$ is in general non-zero. But an
appropriate Girsanov transformation applied to the measure $P$ will achieve
a reduction to the zero-drift case (cf. {\cite[\S\ I.13, eqn. (13.9)]{RogW}}%
), at the cost, however, of an additional exponential factor in (A.3):%
\[
P(\max_{u\in \lbrack 0,T-t]}\{\mu ^{\ast }u+\sigma ^{\ast }W_{u}^{\ast
}\}\geq A_{t}^{\ast })=1,\leqno{\rm (A7a)}
\]%
unless $A_{t}^{\ast }>0,$ in which case%
\[
\leqno{\rm (A7b)}\qquad P(\max_{u\in \lbrack 0,T-t]}\{\mu ^{\ast }u+\sigma
^{\ast }W_{u}^{\ast }\}\geq A_{t}^{\ast })=
\]%
\[
=\frac{1}{2}\mathrm{Erfc}\left( \frac{A_{t}^{\ast }-(T-t)\mu ^{\ast }}{%
\sigma ^{\ast }\sqrt{2(T-t)}}\right) +\frac{1}{2}\exp \left( \frac{2\mu
^{\ast }A_{t}^{\ast }}{(\sigma ^{\ast })^{2}}\right) \mathrm{Erfc}\left(
\frac{A_{t}^{\ast }+(T-t)\mu ^{\ast }}{\sigma ^{\ast }\sqrt{2(T-t)}}\right) ,
\]%
with $\mathrm{Erfc}\,(z):=(2/\sqrt{\pi })\int_{[z,\infty )}\exp (-w^{2})\,dw$%
, for any complex $z$, the complementary error function. This result can be
established, mutatis mutandis, along the lines of {\cite[Lemma~A.18.2,
p.~617seq]{MusR}}); for a proof by a reduction to this result, start from
the equality
\[
P(\max_{u\in \lbrack 0,T-t]}\{\mu ^{\ast }u+\sigma ^{\ast }W_{u}^{\ast
}\}\geq A_{t}^{\ast })=1-P(\max_{u\in \lbrack 0,T-t]}\{\mu ^{\ast }u+\sigma
^{\ast }W_{u}^{\ast }\}\leq A_{t}^{\ast });
\]%
on the right-hand side we have from {\cite[eq. (A.85)]{MusR}} the equality%
\[
P(\max_{u\in \lbrack 0,T-t]}\{\mu ^{\ast }u+\sigma ^{\ast }W_{u}^{\ast
}\}\leq A_{t}^{\ast })=
\]%
\[
=N\left( \frac{A_{t}^{\ast }-(T-t)\mu ^{\ast }}{\sigma ^{\ast }\sqrt{T-t}}%
\right) -\exp \left( 2\frac{\mu ^{\ast }A_{t}^{\ast }}{(\sigma ^{\ast })^{2}}%
\right) N\left( -\frac{A_{t}^{\ast }+(T-t)\mu ^{\ast }}{\sigma ^{\ast }\sqrt{%
T-t}}\right) ,
\]%
if $A_{t}^{\ast }\geq 0$, but otherwise this probability is $0$; then
successively use the identities $1=N(\xi )+N(-\xi )$ and $N(\xi )=(1/2)%
\mathrm{Erfc}{(-\xi /\sqrt{2})}$ to arrive at (A7a,b).

Formulas for the tails of the running-minimum expressions of (A.4) are a
consequence of (A.7a,b). For this start by passing to the complementary
probability
\[
P(\min_{u\in \lbrack 0,T-t]}\{\mu ^{\ast }u+\sigma ^{\ast }W_{u}^{\ast
}\}\geq A_{t}^{\ast })=1-P(\min_{u\in \lbrack 0,T-t]}\{\mu ^{\ast }u+\sigma
^{\ast }W_{u}^{\ast }\}\leq A_{t}^{\ast });
\]%
now use (A.6) to translate the right-hand side in terms of probabilities for
the running maximum, and apply (A.7a,b) to obtain
\[
P(\min_{u\in \lbrack 0,T-t]}\{\mu ^{\ast }u+\sigma ^{\ast }W_{u}^{\ast
}\}\geq A_{t}^{\ast })=\leqno{(A.8)}
\]%
\[
=\frac{1}{2}\mathrm{Erfc}\left( +\frac{A_{t}^{\ast }-(T-t)\mu ^{\ast }}{%
\sigma ^{\ast }\sqrt{2(T-t)}}\right) -\frac{1}{2}\exp \left( \frac{2\mu
^{\ast }A_{t}^{\ast }}{(\sigma ^{\ast })^{2}}\right) \mathrm{Erfc}\left( -%
\frac{A_{t}^{\ast }+(T-t)\mu ^{\ast }}{\sigma ^{\ast }\sqrt{2(T-t)}}\right) ,
\]%
unless $A_{t}^{\ast }\geq 0$, in which case the probability is equal to $0$;
to obtain the first summand here use the identity $2=\mathrm{Erfc}(\xi )$+%
\textrm{Erfc}$(-\xi )$ with
\[
\xi :=\left( A_{t}^{\ast }-(T-t)\mu ^{\ast }\right) /\left( \sigma ^{\ast }%
\sqrt{2(T-t)}\right) .
\]

\bigskip

\bigskip

\bigskip

\noindent Department of Accounting, Bocconi University, Via Roberto Sarfatti
25, 20100 Milano, Italy; miles.gietzmann@unibocconi.it\newline
Mathematics Department, London School of Economics, Houghton Street, London
WC2A 2AE; A.J.Ostaszewski@lse.ac.uk\newline
Keplerstrasse 30, D-69469 Weinheim, Germany; mhgsch@gmail.com\newpage

\end{document}